\documentclass[A4,12pt]{article}

\usepackage[english]{babel}
\usepackage[utf8]{inputenc}
\usepackage{algorithm}
\usepackage[noend]{algpseudocode}
\usepackage{enumitem}
\usepackage{pf2}
\usepackage[table]{xcolor}

\newtheorem{theorem}{Theorem}
\newtheorem{lemma}[theorem]{Lemma}

\newcommand{\msg}[4]{$\langle$\texttt{#1}, $#2$, $#3$, $#4$$\rangle$}
\newcommand{\preprepare}[3]{\msg{PRE-PREPARE}{#1}{#2}{#3}}
\newcommand{\prepare}[3]{\msg{PREPARE}{#1}{#2}{#3}}
\newcommand{\commit}[3]{\msg{COMMIT}{#1}{#2}{#3}}
\newcommand{\roundchange}[4]{$\langle$\texttt{ROUND-CHANGE}, $#1$, $#2$, $#3$, $#4$$\rangle$}

\newcommand{\quorum}{$\lfloor\frac{n+f}{2}\rfloor+1$ }
\newcommand{\belowquorum}{$\lfloor\frac{n+f}{2}\rfloor$ }

\begin{document}

\title{The Istanbul BFT Consensus Algorithm}
\author{
	Henrique Moniz\\
	\textit{Quorum Engineering}\\
	%\normalsize{www.goquorum.com}
}

\maketitle

\abstract{
This paper presents IBFT, a simple and elegant Byzantine fault-tolerant consensus algorithm that is used to implement state machine replication in the \emph{Quorum} blockchain. IBFT assumes a partially synchronous communication model, where safety does not depend on any timing assumptions and only liveness depends on periods of synchrony. The algorithm is deterministic, leader-based, and optimally resilient - tolerating $f$ faulty processes out of $n$, where $n \geq 3f+1$. During periods of good communication, IBFT achieves termination in three message delays and has $O(n^2)$ total communication complexity.
}

\section{Introduction}

In this paper, we present \emph{Istanbul BFT} (IBFT), a Byzantine fault-tolerant (BFT) consensus algorithm that is used for implementing state-machine replication in the \emph{Quorum} blockchain. Quorum is an open source permissioned blockchain platform. It is based on Ethereum and designed for enterprise applications.

IBFT was initially proposed informally in EIP-650~\cite{lin17istanbul}. That first proposal had safety issues where two correct processes could decide different values. A revision addressed these safety issues but introduced liveness issues where some executions could lead to a deadlock~\cite{chuan18tolerate, saltini19correctness}. This paper offers a precise and \emph{correct} description of the algorithm along with correctness proofs, solving both the liveness and safety issues of earlier versions. % and a TLA+ specification.
Saltini et al. provide an alternative solution to the correctness issues in the initial IBFT proposal by applying the PBFT protocol~\cite{castro99byzantine} to a blockchain application~\cite{saltini19ibft20}.

IBFT belongs to a class of BFT algorithms that assume a partially synchronous communication model~\cite{dwork88partial}. Under this model messages can be arbitrarily delayed, but the system is characterized by also having periods of good communication in which messages are timely delivered. The IBFT algorithm is deterministic, leader-based, and optimally resilient - tolerating $f$ faulty processes out of $n$, where $n \geq 3f+1$ \cite{pease80agreement}. During periods of good communication, IBFT achieves termination in three message delays and has a communication complexity of $O(n^2)$.

%We also introduce the \emph{LinearBFT} algorithm, a transformation to IBFT that achieves $O(n)$ communication complexity during both normal operation and view changes at a small latency penalty compared to IBFT - from 3 to 5 message delays.
%We also introduce a variant of IBFT, dubbed \emph{LinearIBFT}, which achieves $O(n)$ communication complexity during both normal operation and view changes at a small latency penalty compared to IBFT - from 3 to 5 message delays.

IBFT follows from the line of work started with PBFT~\cite{castro99byzantine}, which was the first practical algorithm designed for this model. More recently, a number of consensus algorithms have been proposed under the same class with the aim of being used specifically in blockchain systems~\cite{buchman18latest, buterin17casper, gueta19sbft, yin19hotstuff}.

%Each of these algorithms occupies a different point in the design space, resulting in different performance characteristics. 
Table~\ref{table:performance} shows the communication complexity and latency to reach a decision for different algorithms under this class. The communication complexity is the asymptotic upper bound on the total number of bits exchanged by the algorithm. The latency is the number of message delays incurred by a correct process until termination. Both metrics pertain for a period of synchrony in which messages are timely delivered.

\begin{table}
	\label{table:performance}
	\centering
	\rowcolors{3}{gray!10}{white}
	\begin{tabular}{l||cc||c}
		& \multicolumn{2}{c||}{\textbf{Communication Complexity}} & \textbf{Latency}\\ 
		%& Comm. Complexity & Comm. Complexity & Message \\
		& Normal Case & View Change & Message Delays \\
		\hline
		DLS~\cite{dwork88partial} & $O(n^4)$ & $O(n^4)$ & $O(n)$ \\
		PBFT~\cite{castro99byzantine} & $O(n^2)$ & $O(n^3)$ & 3 \\
		Zyzzyva~\cite{kotla07zyzzyva} & $O(n)$ & $O(n^3)$ & 1 / $3^*$ \\
		Spinning~\cite{veronese09spin} & $O(n^2)$ & $O(n^3)$ & 3 \\
		SBFT~\cite{gueta19sbft} & $O(cn)^\dagger$ & $O(n^2)$ & 5 / 7 \\
		HotStuff~\cite{yin19hotstuff} & $O(n)$ & $O(n)$ & 8 \\
		IBFT (Sec. \ref{section:algorithm}) & $O(n^2)$ & $O(n^2)$ & 3 \\
		%LinearIBFT (Sec. \ref{section:linear}) & $O(n)$ & $O(n)$ & 5 \\
	\end{tabular}
	\caption{\small Performance of related algorithms when communication is timely. Message delays for dual-mode protocols are shown as $x$ / $y$ where $x$ is for the optimistic environment and $y$ otherwise.}
	\raggedright	
	\small{
		$^*$Zyzzyva requires, for the slow path, waiting for the maximum network delay $\Delta$ in addition to the 3 message delays. Additionally, its fast path has known correctness issues~\cite{abraham17revisiting}. \\
		$^\dagger$$O(fn)$ if applying the recommended heuristic. See Section~\ref{section:related-work}} for context.
\end{table}

%As we can see, our algorithms occupy interesting points in the design space. 
IBFT minimizes worst-case latency (i.e., the number of messages delays to termination without optimistic assumptions about the environment). It achieves termination in three message delays - matching PBFT, Zyzzyva, and Spinning - while improving on the communication complexity of these algorithms by exchanging a quadratic number of messages. HotStuff achieves linear communication complexity but it does so at a high cost in latency.

%LinearIBFT achieves linear communication complexity during both normal operation and view changes, matching HotStuff, while improving the latency from 8 to 5 message delays.

The remainder of the paper is organized as follows. Section~\ref{section:related-work} discusses the related work. Section~\ref{section:system-model} defines the system model. Section~\ref{section:algorithm} describes the IBFT algorithm in detail.
%Section 4 summarizes the differences between this design and other proposals, particularly the original IBFT and the changes proposed in \cite{DBLP:journals/corr/abs-1901-07160}.
Section~\ref{section:correctness} proves the correctness of IBFT. %Section~\ref{section:linear} presents LinearIBFT. 
Finally, Section~\ref{section:conclusion} concludes the paper. %A TLA+ specification of IBFT is presented in Appendix~\ref{section:specification}.

%%%%%%%%%%%%%%%%%%%%%%%%%%%%%%%%%%%%%%%%%%%%%%%%%%%%%%%%%%%%%%%%%%%%%%
\section{Related Work}
\label{section:related-work}

The problems of consensus and state machine replication (SMR) are closely associated with one another. Consensus requires processes in a distributed system to reach agreement on some value~\cite{turek92consensus}. SMR requires agreement on a total order of commands~\cite{lamport78clocks, schneider90tutorial}. When consensus is solvable, so is SMR. As such, consensus is often used as a building block to implement SMR. 

The problem of a distributed system reaching agreement in the presence of Byzantine process failures was first devised by Pease et al.~\cite{lamport82byzantine, pease80agreement}. They also propose solutions for synchronous systems, where there is a known bound on the message transmission delays and the relative speeds of processes.

Of more practical utility are solutions to the problem of consensus in asynchronous systems, where there are no timing assumptions. Fischer et al., however, proved there there is no deterministic solution to consensus in asynchronous system where a single process is allowed to fail~\cite{fischer85impossibility}. There is an abundant body of research dedicated to circumventing this impossibility result using different techniques. The most notable examples being partial synchrony~\cite{dolev87minimal,dwork88partial}, failure detectors~\cite{aguilera00failure,chandra96detectors}, and randomization~\cite{benor83free, chor89randomization, rabin83randomized}.

In this paper, we are concerned with implementing Byzantine fault-tolerant state machine replication in a partially synchronous system. We thus restrict our comparison to protocols that, like IBFT, assume a partially synchronous model, tolerate Byzantine process failures, and are optimally resilient (i.e., $n \geq 3f+1$).

The partially synchronous model was introduced by Dwork et al.~\cite{dwork88partial}. Along with it, they also proposed a Byzantine fault-tolerant consensus algorithm (DLS), which, although inefficient, proved that the problem has a solution in partially synchronous systems.

The PBFT algorithm was the first to provide a correct solution for state machine replication with Byzantine faults in a partially synchronous system where safety does not depend on timing assumptions~\cite{castro99byzantine}. This work was seminal in that it inspired a long line of algorithms that explore the design space within the same partially synchronous model. 

The Zyzzyva algorithm introduced the idea of using speculative execution to improve performance~\cite{kotla07zyzzyva}. Replicas optimistically adopt the order proposed by the primary and delegate to the clients the detection of inconsistencies, which help replicas resolve their state to one that is consistent with a total ordering of requests. Clement et al. later showed that this approach suffers from significant performance problems if even a single malicious client is present in the system~\cite{clement09making}. As a solution, they propose a new algorithm - Aardvark - that makes more robust design decisions at a cost of best-case performance. More recently, Abraham et al. also demonstrated safety violations in the optimistic execution of Zyzzyva~\cite{abraham17revisiting}.

Spinning was the first algorithm to use a rotating leader replica~\cite{veronese09spin}, a concept later applied by many blockchain-motivated BFT algorithms. The leader replica is changed after every request execution instead of only when it is suspected to have failed.

More recently, we have seen a new wave of protocols that are motivated by their application to blockchain systems. We highlight Tendermint, SBFT, and HotStuff.

Tendermint is a protocol whose main novelty is that it does not have a separate round (i.e., view) change algorithm~\cite{buchman18latest}. Replicas change to a new round $r+1$ as part of the normal operation by reaching a decision on round $r$. Similarly to Spinning, this allows for leader rotation as part of the normal operation and not just when a leader is suspected to be faulty. The main drawback of Tendermint is that even if a round has timely communication and an honest leader, it does not guarantee that a decision will be reached within that round. This is because if a correct process is locked on a block $b$ that is not the one being proposed, then the algorithm needs to keep advancing the view until it reaches one where the leader proposes $b$. This results in a total communication complexity of $O(n^3)$ and latency of $O(n)$.

SBFT is a dual-mode protocol, employing a faster optimistic protocol - inspired by Zyzzyva - when there are no faulty replicas and the system is synchronous, and a slower fallback protocol - similar to PBFT - otherwise~\cite{gueta19sbft}. SBFT makes use of the concept of a \emph{collector}. During a communication round, each replica, instead of broadcasting its message, sends it to a designated replica that aggregates the messages from all replicas into a single message and broadcasts it. Messages are signed using threshold signatures, which allow for the aggregated message to have a constant size. Since a single slow or failed collector would be sufficient to make the system switch to the fallback, slower protocol, SBFT allows the optimistic protocol to tolerate a parameterized number $c$ of slow or failed replicas out of $n = 3f+2c+1$. Thus, for any $c > 0$, the algorithm fails to achieve optimal resiliency. SBFT has $O(cn)$ communication complexity during the normal case and $O(n^2)$ complexity during view changes. If $c$ is a constant, then this results in linear complexity during the normal case. It is unlikely, however, for $c$ to remain a constant value as a system scales. The authors recommend $c \leq f/8$ as a good heuristic, in which case the communication complexity would be $O(fn)$.

HotStuff is another protocol that employs the concept of a collector combined with threshold signatures to reduce communication complexity~\cite{yin19hotstuff}. Unlike SBT, however, it only uses the primary replica as the collector. Like Tendermint, HotStuff does not employ a separate view change protocol. Instead, the view is advanced as part of the normal execution. This allows it to achieve $O(n$) communication. The trade-off is a higher number of message delays to reach a decision.

%%%%%%%%%%%%%%%%%%%%%%%%%%%%%%%%%%%%%%%%%%%%%%%%%%%%%%%%%%%%%%%%%%%%%%
\section{System Model}
\label{section:system-model}
The system is composed by a known set of $n$ processes $\Pi = \{p_1, p_2, ..., p_n\}$. A process that follows the algorithm is said to be \emph{correct}. Otherwise, it is said to be \emph{faulty}. A faulty process can behave in an arbitrary (i.e., Byzantine) way, including sending purposely wrong messages with the intent to obviate the correct execution of the algorithm. The number of faulty processes is constrained to $f$, such that $n \geq 3f+1$

We assume a partially synchronous system, where there is an unknown upper bound $\Delta$ on execution and communication delays that holds after an unknown period of time called \emph{global stabilization time} (or GST, for short)~\cite{dwork88partial}.

Processes communicate by sending messages over a network. Any message sent before GST can be arbitrarily delayed or lost by the network. Reaching GST, we assume the following: a message sent by a correct process at some time $t$, such that $t \geq GST$, is guaranteed to be delivered by all correct processes by time $t + \Delta$. For convenience sake, each process has access to a broadcast primitive that sends a message to every process in $\Pi$, including itself.

%We should note that GST-2 holds regardless of whether the broadcaster is correct or faulty, and it can be achieved, for instance, by point-to-point channels with some message retry mechanism or by using an epidemic dissemination algorithm~\cite{eugster04epidemic}. This property also reflects the nature of some systems that assume some form of overlay network as in the case of the Ethereum network.

%In IBFT, messages of type \texttt{PRE-PREPARE} require a special form validation in order to be accepted by correct a process. This is detailed in Section~\ref{section:validation} below. This validation is performed by the accepting process observing enough messages of type \texttt{ROUND-CHANGE} and/or \texttt{PREPARE} that justify the value carried by the \texttt{PRE-PREPARE} message. This validation can be \emph{implicit}, where a process waits to receive enough messages that validate the \texttt{PRE-PREPARE} message, or \emph{explicit}, where the justifying messages are piggybacked in the \texttt{PRE-PREPARE} messages.

%It is up for the system architect to decide whether to use implicit or explicit validation. Using implicit validation, however, requires assuming a second property:

%[TODO: CLARIFY GST-2 AND RELATIONSHIP WITH VALIDATION]

\section{The IBFT Algorithm}
\label{section:algorithm}

We present the IBFT pseudocode in algorithms \ref{alg:ibft-variables} to \ref{alg:ibft-validation}. The pseudocode depicts how a consensus instance identified by $\lambda$ is executed by a correct process $p_i$. Algorithm \ref{alg:ibft-variables} has the constants, state variables, and the \textproc{Start} procedure, which starts a consensus instance $\lambda$ on $p_i$. Algorithm \ref{alg:ibft-normal} describes the normal case operation of IBFT, which happens during periods where the leader is correct and messages are timely delivered. Algorithm \ref{alg:ibft-rc} details how round changes are performed, for when a leader is suspected to be faulty. Finally, Algorithm \ref{alg:ibft-validation} has the predicates used for message justification, which guarantees correctness during round changes by ensuring that only safe values are proposed. %The logic in the algorithms \ref{alg:ibft-normal} and \ref{alg:ibft-rc} is expressed as a set of event-driven \textsl{upon} rules that are triggered when some condition is met.

%%%%%%%%%%%%%%%%%%%%%% Variables, etc. %%%%%%%%%%%%%%%%%%%%%%
\begin{algorithm}
\caption{IBFT pseudocode for process $p_i$: constants, state variables, and ancillary procedures}\label{alg:euclid}
\label{alg:ibft-variables}
\begin{algorithmic}[1]
\footnotesize

\algdef{SN}[constants]{Constants}{EndConstants}
{\textbf{constants:}}

\algdef{SN}[variables]{Variables}{EndVariables}
{\textbf{state variables:}}

\algdef{SN}[timer]{Timer}{EndTimer}
{\textbf{timer:}}

\Statex
\Constants\label{constants}
\State $p_i$ \Comment{The identifier of the process}
\EndConstants\label{end-constants}

\Statex
\Variables\label{variables}
\State $\lambda_i$ \Comment{The identifier of the consensus instance}
\State $r_i$ \Comment{The current round}
\State $pr_i$ \Comment{The round at which the process has prepared}
\State $pv_i$ \Comment{The value for which the process has prepared}
\State $inputValue_i$ \Comment{The value passed as input to this instance}
\EndVariables\label{end-variables}

\Statex
\Timer\label{timer}
\State $timer_i$
\EndTimer\label{end-timer}

\Statex
\Procedure{Start}{$\lambda$, $value$}\label{procedure:start}
	\State $\lambda_i \gets \lambda$
	\State $r_i \gets 1$
	\State $pr_i \gets \bot$
	\State $pv_i \gets \bot$
	\State $inputValue_i \gets value$
	\If{\textsc{leader}($h_i$, $r_i$) $= p_i$}
		\State broadcast \preprepare{\lambda_i}{r_i}{inputValue_i} message
	\EndIf
	\State set $timer_i$ to \texttt{running} and expire after $t(r_i)$
\EndProcedure\label{end-procedure:start}

%\Statex
%\Procedure{RestartTimer()}{}\label{procedure:restart-timer}
%	\State set $timer_i$ to \texttt{running} and expire after $t(r_i)$
%\EndProcedure

\end{algorithmic}
\end{algorithm}

%%%%%%%%%%%%%%%%%%%%%% Upon rules %%%%%%%%%%%%%%%%%%%%%%
\begin{algorithm}
\caption{IBFT pseudocode for process $p_i$: normal case operation}
\label{alg:ibft-normal}
\begin{algorithmic}[1]
\footnotesize

\algdef{SN}[upon]{Upon}{EndUpon}
	[1][]{\textbf{upon} #1 \textbf{do}}

\Statex
\Upon [receiving a valid \preprepare{\lambda_i}{r_i}{value} message $m$ from \textproc{leader($\lambda_i$, $round$)} such that \textproc{JustifyPrePrepare($m$)}]\label{upon:normal-preprepare}
	\State set $timer_i$ to \texttt{running} and expire after $t(r_i)$
	\State broadcast \prepare{\lambda_i}{r_i}{value}
\EndUpon\label{end-upon:normal-preprepare}

\Statex
\Upon[receiving a quorum of valid \prepare{\lambda_i}{r_i}{value} messages]\label{upon:normal-prepare}
	\State $pr_i$ $\gets$ $r_i$
	\State $pv_i$ $\gets$ $value$
	\State broadcast \commit{\lambda_i}{r_i}{value}
\EndUpon\label{end-upon:normal-prepare}

\Statex
\Upon[receiving a quorum $Q_{commit}$ of valid \commit{\lambda_i}{round}{value} messages]\label{upon:normal-commit}
	\State set $timer_i$ to \texttt{stopped}
	\State \textproc{Decide($\lambda_i, value, Q_{commit}$)}
\EndUpon\label{end-upon:normal-commit}

\end{algorithmic}
\end{algorithm}
%%%%%%%%%%%%%%%%%%%%%%%%%%%%%%%%%%%%%%%%%%%%%%%%%%%%%%%%%%%%%

%%%%%%%%%%%%%%%%%%%%%% Round Changes %%%%%%%%%%%%%%%%%%%%%%
\begin{algorithm}
\caption{IBFT pseudocode for process $p_i$: round changes}
\label{alg:ibft-rc}
\begin{algorithmic}[1]
\footnotesize

\algdef{SN}[upon]{Upon}{EndUpon}
	[1][]{\textbf{upon} #1 \textbf{do}}

\Statex
\Upon[$timer_i$ is \texttt{expired}]\label{upon:rc-timer-expired}
	\State $r_i \gets r_i + 1$
	\State set $timer_i$ to \texttt{running} and expire after $t(r_i)$
	\State broadcast \roundchange{\lambda_i}{r_i}{pr_i}{pv_i}
\EndUpon\label{end-upon:rc-timer-expired}

\Statex
\Upon[receiving a set $F_{rc}$ of $f+1$ valid \roundchange{\lambda_i}{r_j}{_-}{_-} messages such that $\forall$\roundchange{\lambda_i}{r_j}{_-}{_-} $\in F_{rc}: r_j > r_i$]\label{upon:rc-set-f}
	\State let \roundchange{h_i}{r_{min}}{_-}{_-} $\in F_{rc}$ such that:\\
		\qquad\qquad $\forall$\roundchange{\lambda_i}{r_j}{_-}{_-} $\in F_{rc}: r_{min} \leq r_j$
	\State $r_i \gets r_{min}$
	\State set $timer_i$ to \texttt{running} and expire after $t(r_i)$	
	\State broadcast \roundchange{\lambda_i}{r_i}{pr_i}{pv_i}
\EndUpon\label{end-upon:rc-set-f}

%\Statex
%\Statex {\scriptsize $\triangleright$ Note that a valid \texttt{PRE-PREPARE} message for $r_i > 1$ can also trigger this rule because it implies a quorum of \texttt{ROUND-CHANGE} messages.}
%\Upon[receiving a quorum  of valid \roundchange{\lambda_i}{r_i}{_-}{_-} messages]% $\land$ $timer_i$ is not \texttt{running}]
%\label{upon:rc-timer-restart}
%	\State \textproc{RestartTimer()}
%\EndUpon\label{end-upon:rc-timer-restart}

\Statex
\Upon[receiving a quorum $Q_{rc}$ of valid \roundchange{\lambda_i}{r_i}{_-}{_-} messages such that $\textproc{leader($\lambda_i$, $r_i$)} = p_i$ $\land$ \textproc{JustifyRoundChange($Q_{rc}$)}]\label{upon:rc-broadcast-preprepare}
	\If{$\textproc{HighestPrepared($Q_{rc}$)} \neq \bot$}\label{choose-value}
		\State let $v$ such that $(_-, v) = \textproc{HighestPrepared($Q_{rc}$)})$
	\Else
		\State let $v$ such that $v = inputValue_i$
	\EndIf\label{end-choose-value}
	\State broadcast \preprepare{\lambda_i}{r_i}{v}
\EndUpon\label{end-upon:rc-broadcast-preprepare}

\Statex
\Statex {\scriptsize $\triangleright$ We can omit this if we assume some mechanism external to the consensus algorithm that ensures synchronization of decided values.}
\Upon[receiving a valid \roundchange{\lambda_i}{_-}{_-}{_-} message from  $p_j$ $\land$ $p_i$ has decided by calling \textproc{Decide($\lambda_i, _-, Q_{commit}$)}]\label{upon:rc-synchronization}
	\State send $Q_{commit}$ to process $p_j$
\EndUpon\label{end-upon:rc-synchronization}

\end{algorithmic}
\end{algorithm}
%%%%%%%%%%%%%%%%%%%%%%%%%%%%%%%%%%%%%%%%%%%%%%%%%%%%%%%%%%%%%%%%

%%%%%%%%%%%%%%%%%%%%%% Message Validation %%%%%%%%%%%%%%%%%%%%%%
\begin{algorithm}
\caption{IBFT pseudocode for process $p_i$: message justification}
\label{alg:ibft-validation}
\begin{algorithmic}[1]
\footnotesize

\algdef{SN}[predicate]{Predicate}{EndPredicate}
	[1][]{\textbf{predicate} #1}	

\Statex
\Predicate{\textproc{JustifyRoundChange($Q_{rc}$)}}\label{predicate:justifyroundchange}
	\State \textbf{return}
	\Statex \qquad\quad $\forall$\roundchange{\lambda_i}{r_i}{pr_j}{pv_j} $\in Q_{rc}: pr_j = \bot \land pv_j = \bot$
	\Statex \qquad\quad $\lor$ received a quorum of valid \prepare{\lambda_i}{pr}{pv} messages such that:
	\Statex \qquad\qquad\quad $(pr,pv) = \textproc{HighestPrepared}(Q_{rc})$
\EndPredicate\label{end-predicate:justifyroundchange}

\Statex
\Predicate{\textproc{JustifyPrePrepare(\preprepare{\lambda_i}{round}{value})}}\label{predicate:justifypreprepare}
	\State \textbf{return}
	\Statex \qquad\quad $round = 1$
	\Statex \qquad\quad $\lor$ received a quorum $Q_{rc}$ of valid \roundchange{\lambda_i}{round}{pr_j}{pv_j} messages
	\Statex \qquad\quad\quad such that:
	\Statex \qquad\qquad\quad $\forall$\roundchange{\lambda_i}{round}{pr_j}{pv_j} $\in Q_{rc}: pr_j = \bot \land pr_j = \bot$
	\Statex \qquad\qquad\quad $\lor$ received a quorum of valid \prepare{\lambda_i}{pr}{value} messages such that:
	\Statex \qquad\qquad\qquad\quad $(pr,value) =$ \textproc{HighestPrepared($Q_{rc}$)}
\EndPredicate\label{end-predicate:justifypreprepare}

\Statex
\Statex {\scriptsize $\triangleright$ Helper function that returns a tuple $(pr, pv)$ where $pr$ and $pv$ are, respectively, the prepared round and the prepared value of the \texttt{ROUND-CHANGE} message in $Q_{rc}$ with the highest prepared round}
\Function{HighestPrepared}{$Q_{rc}$}\label{function:highestprepared}
	\State \textbf{return} $(pr, pv)$ such that:
	\Statex \qquad\quad $\exists$\roundchange{\lambda_i}{round}{pr}{pv} $\in Q_{rc} :$	
	\Statex \qquad\qquad\ \ \ $\forall$\roundchange{\lambda_i}{round}{pr_j}{pv_j} $\in Q_{rc} : pr_j = \bot \lor pr \geq pr_j$
\EndFunction\label{end-function:highestprepared}

\end{algorithmic}
\end{algorithm}
%%%%%%%%%%%%%%%%%%%%%%%%%%%%%%%%%%%%%%%%%%%%%%%%%%%%%%%%%%%%%

%For simplicity of presentation some details are omitted from the pseudocode. Below, we provide a complete explanation of the algorithm that includes those details.

\subsection{Preliminaries}
\label{section:preliminaries}

The IBFT algorithm solves the consensus problem, where all correct processes need to decide on some common value. More formally, each  instance of IBFT guarantees the following properties:

\begin{description}
	\item[Agreement.] If a correct process decides some value $v$, then no correct process decides a value $v'$ such that $v' \neq v$.
	\item[Validity.] Given an externally provided predicate $\beta$, if a correct process decides some value $v$, then $\beta(v)$ is true.
	\item[Termination.]	Every correct process eventually decides.
\end{description}

Our validation condition deserves further explanation. It uses the notion of \emph{external validity}, originally proposed by Cachin et al. \cite{cachin01secure}. The application calling the algorithm provides an arbitrary predicate $\beta$ whose purpose is to ensure that the decided value is acceptable within the context of the application. For instance, a blockchain implementation might want to ensure that the decided value is a block containing legitimate transactions.

Each instance $\lambda$ of the algorithm proceeds in rounds. During each round, one of the processes acts has a leader that tries to drive the execution to a common decision by proposing a value. During a \emph{good} round, where communication is timely and the leader is not faulty (i.e., after GST), the algorithm guarantees that all correct processes will reach a decision. There is a function \textsc{leader}($\lambda$, $round$) that identifies the leader. This function can be any deterministic mapping from $\lambda$ and $round$ to the identifier of a process as long as it allows $f+1$ processes to eventually assume the leader role.

\paragraph{Messages.}
Messages are represented as tuples enclosed in angle brackets. There are four types of messages: \texttt{PRE-PREPARE}, \texttt{PREPARE}, \texttt{COMMIT}, and \texttt{ROUND-CHANGE}. The first three types are of the form $\langle$\texttt{message-type}, $\lambda$, $r$, $value$$\rangle$ - where $\lambda$ is the consensus instance, $r$ is the round, and $value$ is the proposal value - and comprise the normal case operation of the algorithm. A \texttt{ROUND-CHANGE} message is of the form \roundchange{\lambda}{r}{pr}{pv} - where $\lambda$ is the consensus instance, $r$ is the round, $pr$ is the prepared round, and $pv$ is the prepared value - and is used to ensure progress when the current leader is suspected to have failed or the communication is not timely. 
% The first three types are of the form $\langle$\texttt{message-type}, $\lambda$, $round$, $value$$\rangle$ and comprise the normal case operation of the algorithm. The \texttt{ROUND-CHANGE} message is of the form $\langle$\texttt{ROUND-CHANGE}, $height$, $round$, $preparedRound$, $preparedValue$$\rangle$ and is used ensure progress when the current leader is suspected to have failed or the communication is not timely.

\paragraph{State.}
The algorithm state is composed of five variables: the consensus instance $\lambda_i$, the round $r_i$, the prepared round $pr_i$, and the prepared value $pv_i$, and the input value $inputValue_i$.

The variable $\lambda_i$ identifies the instance of the consensus algorithm being executed. It is set upon a call to the \textsc{Start} procedure and it never changes throughout the execution. When the algorithm is used to implement state machine replication, the instances are numbered in a total order that determines the execution of commands. In a blockchain system, $\lambda_i$ can correspond to the block number.

The variable $r_i$ identifies the round in which process $p_i$ is currently on and it starts at 1.

The prepared round $pr_i$ and prepared value $pv_i$ variables are, respectively, the highest round and the corresponding value for which $p_i$ has \emph{prepared}. During the execution of a consensus instance $\lambda$, we say that a process $p_i$ has prepared for a round $r$ and a value $v$ if it receives a \emph{quorum of valid messages}\footnote{We explain in the validation paragraph below what this entails.} of the form \prepare{\lambda}{round}{value}. These variables are initialized with a default value $\bot$, which means that $p_i$ has not prepared yet. This mechanism is essential for the safety of the algorithm and we explain how it works in Section~\ref{section:justification}.

Finally, $inputValue_i$ is simply the value passed as input to process $p_i$, which is saved in this variable.

\paragraph{Timer.} In addition to the state variables, each correct process $p_i$ also maintains a timer represented by $timer_i$, which is used to trigger a round change when the algorithm does not sufficiently progress. The timer can be in one of two states: \texttt{running} or \texttt{expired}. When set to \texttt{running}, it is also set a time $t(r_i)$, which is an exponential function of the round number $r_i$, after which the state changes to \texttt{expired}.

\paragraph{Upon rules.}
The main logic of the algorithm is expressed as a set of event-driven \textsl{upon} rules that are triggered when some condition is met. These are found exclusively in Algorithms \ref{alg:ibft-normal} and \ref{alg:ibft-rc}. While not explicit in the pseudocode, we impose the restriction that, within an instance of the algorithm, each upon rule is triggered at most once for any round $r_i$. The only exception is the last rule on Algorithm \ref{alg:ibft-rc} (line~\ref{upon:rc-synchronization}), which can be triggered any number of times.

A condition necessary to trigger many upon rules is receiving a \emph{quorum of valid messages} that match a certain pattern. We say that a process has received a quorum of valid messages if it has received \emph{valid} messages from $\lfloor \frac{n+f}{2} \rfloor + 1$ different processes. For instance, the upon rule in line 3 of Algorithm \ref{alg:ibft-normal} is triggered when process $p_i$ receives $\lfloor \frac{n+f}{2} \rfloor + 1$ valid messages from different processes that have type \texttt{PREPARE}, match instance $\lambda_i$ and round $r_i$, and have the same $value$. 

\paragraph{Validation.}
A correct process only accepts a message if it considers it to be \emph{valid}. To be valid, a message must carry some proof of integrity and authentication of its sender such as a digital signature. The external validity predicate $\beta$ must also be true for the value carried by the message. Furthermore, a \roundchange{\lambda}{r}{pr}{pv} message to be valid needs for the prepared round to be smaller than the round, i.e., $pr < r$.

%Furthermore, \texttt{PRE-PREPARE} messages (for any round higher than $1$) and \texttt{ROUND-CHANGE} messages require additional validation, which is expressed, respectively, by the predicates \textsc{ValidatePrePrepare} and \textsc{ValidateRoundChange} in Algorithm \ref{alg:ibft-validation}. This is to ensure the correctness of round changes. We defer their explanation to Section \ref{section:round-changes} where we describe round changes in detail.

\subsection{Normal case operation.}

We now explain how the algorithm works during normal case operation, i.e., in some round where communication is timely and the leader is correct. 

For any correct process $p_i$, an execution of an instance $\lambda$ of the algorithm begins with a call to the \textproc{Start} procedure (Algorithm \ref{alg:ibft-variables}, line~\ref{procedure:start}), which takes as input parameters the instance identifier $\lambda$ and an input \emph{value}. The procedure then initializes the state variables and if $p_i$ is the leader for the current round, it broadcasts a \texttt{PRE-PREPARE} message proposing $inputValue_i$.

The remainder of the normal case operation is expressed in Algorithm \ref{alg:ibft-normal}. Upon receiving a valid \texttt{PRE-PREPARE} message from the leader for instance $\lambda_i$ and current round $r_i$ that carries a \emph{justified}\footnote{This is explained in Section~\ref{section:justification}.} value, a process $p_i$ restarts the timer and broadcasts a \texttt{PREPARE} message (lines \ref{upon:normal-preprepare}-\ref{end-upon:normal-preprepare}).

Upon receiving a quorum of valid \texttt{PREPARE} messages for instance $\lambda_i$ and current round $r_i$ carrying the same $value$, a process $p_i$ updates its prepared round $pr_i$ and prepared value $pv_i$ variables to match the value of the received messages (lines \ref{upon:normal-prepare}-\ref{end-upon:normal-prepare}). We now say that $p_i$ has \emph{prepared} for round $r_i$ and value $value$. It then broadcasts a \texttt{COMMIT} message carrying $value$. 

Finally, upon receiving a quorum of valid \texttt{COMMIT} messages for instance $\lambda_i$ with the same $round$ and $value$, a process $p_i$ decides by calling an externally provided \textproc{Decide} function (lines \ref{upon:normal-commit}-\ref{end-upon:normal-commit}).% and (2) stopping the execution of the algorithm for height $h_i$. Note that the \texttt{COMMIT} message does not carry a round number. This is because, at this point, there are no more algorithm steps to take that depend on the round.

\subsection{Round changes}
\label{section:round-changes}
The previous section explains how the algorithm works under good conditions. The algorithm, however, must be able to tolerate arbitrary periods where communication is untimely or the leader is faulty. Round changes are how the algorithm ensures liveness by allowing different processes to assume the role of a leader. The pseudocode for round changes is expressed in Algorithm \ref{alg:ibft-rc}.

A round change is primarily triggered by the timer, which is available in each process for each consensus instance and represented by $timer_i$. If the algorithm has not made sufficient progress for a process $p_i$ to decide during some round $r_i$, then the timer will eventually expire (Algorithm \ref{alg:ibft-rc}, line \ref{upon:rc-timer-expired}). When this happens, $p_i$ advances to round $r_i+1$ and broadcasts a \texttt{ROUND-CHANGE} message. This message carries the values of the prepared round $pr_i$ and prepared value $pv_i$ variables, which will be used by the new leader to select a value to propose in a \texttt{PRE-PREPARE} message for $r_i + 1$. We explain how this value is selected in Section \ref{section:justification}, which discusses message justification.

The upon rule at line \ref{upon:rc-set-f} is also used for liveness. It is not strictly required in the sense that the algorithm would still be live without it, but it helps ensuring that processes do not wait too long before advancing to a new round. Whenever a process $p_i$ receives a set of $f+1$ valid \texttt{ROUND-CHANGE} messages with any round number higher than its current round $r_i$, it advances to the smallest round within that set and broadcasts a \texttt{ROUND-CHANGE} message.

%Other than the upon rule at line \ref{upon:rc-timer-expired}, the two following upon rules (lines \ref{upon:rc-set-f} and \ref{upon:rc-broadcast-preprepare}) are also used for liveness. They are not strictly required in the sense that the algorithm would still be live without them, but they help in maximizing the time that correct processes spend in the same round, thus improving the chances for a decision to be reached in any particular round. These techniques are similar to the ones employed by PBFT~\cite{castro99byzantine}.

%The upon rule at line \ref{upon:rc-set-f} ensures that processes do not wait too long before advancing to a new round. Whenever a process $p_i$ receives a set of $f+1$ valid \texttt{ROUND-CHANGE} messages with any round number higher than its current round $r_i$, it advances to the smallest round within that set and broadcasts a \texttt{ROUND-CHANGE} message.

%The upon rule at line ?? ensures that processes do not start the timer for a new round too soon by only starting the timer when they receive a quorum of valid \texttt{ROUND-CHANGE} messages. As we will see below, a properly justified \texttt{PRE-PREPARE} message for any round higher than 1 requires a quorum of valid \texttt{ROUND-CHANGE} messages, so this is implied when a correct process broadcasts a \texttt{PRE-PREPARE}.

The upon rule at line \ref{upon:rc-broadcast-preprepare} starts the normal operation of the new round by having the leader broadcast a \texttt{PRE-PREPARE} message. From this point on, the algorithm resumes as specified by Algorithm \ref{alg:ibft-normal}.

Finally, the upon rule at line \ref{upon:rc-synchronization} ensures that any process $p_j$ can catch up to a decision already made by process $p_i$ by having $p_i$ send $p_j$ a quorum of \texttt{COMMIT} messages for $\lambda_i$.

\subsection{Message Justification: Safety Across Rounds}
\label{section:justification}
While round changes ensure liveness, we also need to ensure that the proposal value chosen by the leader of a new round is safe. For this, we rely on the following property:

\begin{itemize}
	%\item If some correct process could have decided a value $v$ at some round $r$, then a \texttt{PRE-PREPARE} message proposing a value $v' \neq v$ for round $r+1$ cannot be accepted by a correct process.

	\item If some correct process could have decided a value $v$ at some round $r < r'$, then $v$ must be the value proposed in a \texttt{PRE-PREPARE} message for round $r'$.		
\end{itemize}

To select a proposal value that guarantees this property, the leader process, upon receiving a quorum $Q_{rc}$ of \texttt{ROUND-CHANGE} messages for round $r'$ (Algorithm \ref{alg:ibft-rc}, line \ref{upon:rc-broadcast-preprepare}), follows the logic expressed in lines \ref{choose-value}-\ref{end-choose-value}. If there is any message in $Q_{rc}$ with prepared round and prepared value not equal to $\bot$, then the leader chooses $v$ as the prepared value $pv$ of the message with the highest prepared round $pr$ among the messages in $Q_{rc}$. Otherwise, if all the messages in $Q_{rc}$ have prepared round and prepared value equal to $\bot$, then the leader chooses the input value that was passed to the \textproc{Start} function and saved in the $inputValue_i$ variable.

Choosing the proposal value in this way ensures safety across rounds. This is because if a correct process decides some value $v$ during a round $r$, then $v$ will be chosen as the proposal value for round $r+1$. To see why, say that a correct process decides $v$ during round $r$. Then, at least $f+1$ correct processes must have prepared for $v$ during round $r'$. This implies that any quorum of \texttt{ROUND-CHANGE} messages will have at least one message with prepared round and prepared value set to $r$ and $v$, respectively, and $v$ will be chosen as the proposal value if the processes follow the algorithm. Since $v$ is the proposal value for round $r+1$, it is not possible for a correct process to decide a different value during round $r+1$. A form of inductive reasoning applies for any subsequent rounds.

While choosing the proposal value in this way ensures safety across rounds, it is possible for a leader that is faulty to deviate from algorithm and send a \texttt{PRE-PREPARE} message with any arbitrary value. Likewise, any faulty process can send a \texttt{ROUND-CHANGE} message with an incorrect prepared round and prepared value to try to deceive a correct leader into choosing an unsafe proposal.

As such, to demonstrate that they are correctly constructed, these messages need a \emph{justification}. A justification is a set of other messages that prove that the message being justified is congruent with the algorithm. For example, a message \roundchange{\lambda}{r}{pr}{pv} sent by a process $p_i$, which  states that $p_i$ prepared for round $pr$ and value $pv$, is only congruent with the algorithm if it exists a quorum of \prepare{\lambda}{pr}{pv} messages. A justification is piggybacked on the message being justified, but it is not part of the message itself.

 The predicates \textproc{JustifyPrePrepare} and \textproc{JustifyRoundChange} in Algorithm \ref{alg:ibft-validation} express the justification logic, which we further explain below.

\paragraph{ROUND-CHANGE justification.}
Messages of the type \texttt{ROUND-CHANGE} are justified together as a set. A correct process $p_i$ considers a quorum $Q_{rc}$ of \texttt{ROUND-CHANGE} messages to be justified if one of the following conditions is true:

\begin{description}
\item[J1] All of the messages in $Q_{rc}$ have prepared round and prepared value equal to $\bot$.
\item[J2] The justification has a quorum of valid \prepare{\lambda_i}{pr}{pv} messages such that \roundchange{\lambda_i}{r}{pr}{pv} is the message with the highest prepared round different than $\bot$ in $Q_{rc}$.
\end{description}

\paragraph{PRE-PREPARE justification.}

A \texttt{PRE-PREPARE} message for round $1$ does not need justification. For any round $r > 1$, a \preprepare{\lambda}{r}{value} message is justified if its justification has a quorum $Q_{rc}$ of \texttt{ROUND-CHANGE} messages such that conditions \textbf{J1} or \textbf{J2} are true, and if \textbf{J2} is true, then $value=pv$.

\section{Correctness}
\label{section:correctness}

In this section we prove the correctness of IBFT. It is organized in three subsections, one for each property of the algorithm as specified in Section~\ref{section:preliminaries} - agreement, validity, and termination. For improved readability, proofs that require longer or more complex arguments are written using the structured proof format proposed by Lamport \cite{lamport12proof}.

%%%% REWRITING STARTS HERE %%%%
\subsection{Agreement}

The proof of the agreement is structured as Lemmas \ref{lemma:prepare-same-round} and \ref{lemma:no-preprepare-following-rounds} and Theorem \ref{theorem:agreement}. Lemma \ref{lemma:prepare-same-round} is used as support to prove agreement during the same round, and Lemma \ref{lemma:no-preprepare-following-rounds} to prove agreement across different rounds. Theorem \ref{theorem:agreement} uses both lemmas to conclude the reasoning.

\paragraph{Agreement during the same round.}

\begin{lemma}
\label{lemma:prepare-same-round}
If some correct process prepares for a value $v$ and round $r$, then no correct process prepares for a value $v'$ and round $r$ such that $v' \neq v$.
\end{lemma}

For this proof we assume that a correct process has prepared for value $v$ at round $r$ and we demonstrate that no correct process can prepare for a different value $v'$ at the same round because there is no quorum of \texttt{PREPARE} messages for $v'$.

\begin{proof}
	\beforePfSpace{1ex, 1ex}
	\afterPfSpace{1ex, 2ex}
	\interStepSpace{1ex}	
	\pflongnumbers
	
	\step{1}{
		A correct process received \quorum valid \prepare{\lambda}{r}{v} messages.
	}
	\begin{proof}
		\pf\ This follows from the assumption. To prepare for value $v$ and round $r$, a correct process must receive \quorum valid \prepare{\lambda}{r}{v} messages.
	\end{proof}

	\step{2}{
		$f+1$ correct processes broadcasted a \prepare{\lambda}{r}{v} message.
	}
	\begin{proof}
		\pf\ This follows from \stepref{1}. If a correct process received \quorum valid \prepare{\lambda}{r}{v} messages, then $f+1$ of those messages must be been broadcasted by correct processes.
	\end{proof}

	\step{3}{
		For any value $v' \neq v$, at most \belowquorum valid \prepare{h}{r}{v'} messages were broadcasted.
	}
	\begin{proof}
		\pf\ We have established in \stepref{2} that $f+1$ correct processes broadcasted a \prepare{\lambda}{r}{v} message. It follows that at most \belowquorum processes could have broadcasted a \prepare{\lambda}{r}{v'} message.
	\end{proof}

	\qedstep
	\begin{proof}
		\pf\ This follows from \stepref{3} since to prepare for a value $v'$ and round $r$, a correct process requires \quorum valid \prepare{\lambda}{r}{v'} messages.	
	\end{proof}
\end{proof}

%%%%%%%%%%%%%%%%%%%%%%%%%%%%%%%%%%%%%%%%%%%%%%%%%%%%%%%%%%%%%%%%%%%%%
\paragraph{Agreement across rounds.}
Lemma \ref{lemma:prepare-same-round} showed that no two correct processes can prepare for different values on the same round. To prove agreement across rounds we will rely on the \ref{lemma:no-preprepare-following-rounds}, which states that that once $f+1$ correct processes prepare for the same round $r$ and value $v$, then a \texttt{PRE-PREPARE} message for the subsequent round $r+1$ must propose $v$ to be justified.

\begin{lemma}
\label{lemma:no-preprepare-following-rounds}
If $f+1$ correct processes prepare for a value $v$ and round $r$, then, for any \preprepare{\lambda}{r'}{v'} message $m$ such that $r' > r$ and $v' \neq v$, \textproc{JustifyPreprepare($m$)} must be false.
\end{lemma}

We prove this lemma by induction on $r'$. In the base case we prove the statement for $r' = r + 1$. In the induction step we prove for $r' > r+1$ and assume that \textproc{JustfifyPreprepare}($m$) is false for any $r''$ such that $r+1 \leq r'' < r'$. 
\\\\
For both cases, the argument is constructed by demonstrating that \textproc{JustifyPreprepare} must be false for any \preprepare{\lambda}{r'}{v'} message $m$. According to the definition of the predicate, \textproc{JustifyPreprepare($m$)} is true if there is a quorum $Q_{rc}$ of \quorum valid \roundchange{\lambda}{r+1}{pr_j}{pv_j} messages such that one of the following two conditions is true, which correspond to J1 and J2 from Section \ref{section:justification}:

\begin{itemize}[leftmargin=6ex]
	\item[(a)] $\forall$\roundchange{\lambda_i}{r+1}{pr_j}{pv_j} $\in Q_{rc}: pr_j = \bot \land pr_j = \bot$
	\item[(b)] There are \quorum valid \prepare{\lambda}{r'}{v'} messages such that $v' \neq v$ and $(r',v') =$ \textproc{HighestPrepared($Q_{rc}$)}
\end{itemize}

We show, both for the base case and for the induction step, that both (a) and (b) are false.
\\
\begin{proof}
  \beforePfSpace{1ex, 1ex}
  \afterPfSpace{1ex, 2ex}
  \interStepSpace{1ex}
  \pflongnumbers

  \step{base}{
    \textproc{Base Case:} $r' = r+1$
    
    \begin{proof}
      \step{base-a}{
        \prove{Condition (a) is false.}
        \begin{proof}
		  \step{base-a-1}{
			Any quorum $Q_{rc}$ includes some message from a correct process that prepared for round $r$ and value $v$.
          }
          \begin{proof}
            \pf\ A quorum $Q_{rc}$ by definition has \quorum messages. Since, by assumption, $f+1$ correct processes prepared for round $r$ and value $v$, it follows that $Q_{rc}$ must include some message from those $f+1$ processes.
          \end{proof}

          \qedstep
          \begin{proof}
            \pf\ By \stepref{base-a-1}, $Q_{rc}$ has some message from a correct process that prepared for round $r$ and value $v$. It follows that $\exists$\roundchange{\lambda_i}{r+1}{pr_j}{pv_j} $\in Q_{rc}:$ $pr_j \neq \bot$ $\land$ $pr_j \neq \bot$.
          \end{proof}
        \end{proof}
      }
      
      \step{base-b}{
        \prove{Condition (b) is false.}
        \vspace{0.5ex}
        For the proof argument, we assume that condition (b) is true and show that it leads to a contradiction.

        \begin{proof}
          \step{base-b-1}{
            Any quorum $Q_{rc}$ includes some message from a correct process that prepared for round $r$ and value $v$.
          }
          \begin{proof}
            \pf\ A quorum $Q_{rc}$ by definition has \quorum messages. Since, by assumption, $f+1$ correct processes prepared for round $r$ and value $v$, it follows that $Q_{rc}$ must include some message from those $f+1$ processes.
          \end{proof}

          \step{base-b-2}{$r' = r$}
          \begin{proof}
            \pf\ Since we assume that condition 2 is true, $r'$ must be the highest prepared round value $pr_j$ in a \texttt{ROUND-CHANGE} message in $Q_{rc}$. By \stepref{base-b-1}, some \texttt{ROUND-CHANGE} message in $Q_{rc}$ has prepared round $r$, which is also the highest prepared round value possible for the message to be valid. It follows that $r' = r$.
					% TODO: Split this proof into multiple assertions.
          \end{proof}
			
          \step{base-b-3}{
            $f+1$ correct processes broadcasted a \prepare{\lambda}{r}{v} message.
          }
          \begin{proof}
            \pf\ This follows from the assumption that $f+1$ correct processes prepared for round $r$ and value $v$.
          \end{proof}

          \step{base-b-4}{
            For any value $v' \neq v$, at most \belowquorum valid \prepare{h}{r'}{v'} messages were broadcasted.
          }
          \begin{proof}
            \pf\ This follows from \stepref{base-b-2} and \stepref{base-b-3}.
          \end{proof}

          \qedstep
          \begin{proof}
            \pf\ The assumption in condition 2 that there are \quorum valid \prepare{\lambda}{r'}{v'} messages is in contradiction with assertion \stepref{base-b-4}.
          \end{proof}
        \end{proof}        
      }
    \end{proof}
  }
  
  \step{induction}{
    \textproc{Induction Step:} $r' > r+1$

    \vspace{1.0ex}    
    We now reason about the induction step. We prove that, for any \preprepare{\lambda}{r'}{v'} message $m$ such that $r' > r+1$ and $v' \neq v$, \textproc{JustifyPreprepare($m$)} is false. For this, we use the induction assumption, which is, for any \preprepare{\lambda}{r''}{v'} message $m'$ such that $r+1 \leq r'' < r'$ and $v' \neq v$, \textproc{JustifyPreprepare($m'$)} is false.
    
    \begin{proof}
      \step{induction-a}{
        \prove{Condition (a) is false.}
      }
      \begin{proof}
        \pf\ The same argument as 1.1 applies here.          
      \end{proof}
      
      \step{induction-b}{
        \prove{Condition (b) is false.}

        \step{induction-b-1}{
          If $(pr, v') =$ \textproc{HighestPrepared($Q_{rc}$)}, then $pr \geq r$. Additionally, $pr < r'$.
        }
        \begin{proof}
          \pf\ By assumption, $f+1$ correct processes prepared for value $v$ and round $r$. This means that any quorum $Q_{rc}$ has some \texttt{ROUND-CHANGE} message with prepared round $p_j$ equal or higher than $r$. As such, for $pr$ to be the highest prepared round value of any message in $Q_{rc}$, we must have $pr \geq r$. The assertion $pr < r'$ follows from the definition of a valid \texttt{ROUND-CHANGE} message.
        \end{proof}  

        \step{induction-b-2}{
          There is no quorum of \prepare{\lambda}{r}{v'} messages such that $v' \neq v$.
        }
        \begin{proof}
          \pf\ This follows from Lemma 1.
        \end{proof}

        \step{induction-b-3}{
          There is no quorum of \prepare{\lambda}{r''}{v'} messages such that $r+1 \leq r'' < r'$ and $v' \neq v$.
        }
        \begin{proof}
          \pf\ A correct process only broadcasts a \prepare{\lambda}{r''}{v'} if it accepts a \preprepare{\lambda}{r''}{v'} message. By assumption, there is no such \texttt{PRE-PREPARE} message. It follows that it is not possible have a quorum of \prepare{\lambda}{r''}{v'} messages
        \end{proof}

        \qedstep
        \begin{proof}
          \pf\ In \stepref{induction-b-1} we established that $r \leq pr < r'$, and in \stepref{induction-b-2} and \stepref{induction-b-3} we established that there is no quorum of \texttt{PREPARE} messages with value $v' \neq v$ for any round from $r$ to $r'-1$.
        \end{proof}
      }
    \end{proof}
  }
\end{proof}

%%%%%%%%%%%%%%%%%%%%%%%%%%%%%%%%%%%%%%%%%%%%%%%%%%%%%%%%%%%%%%%%%%%%%

\paragraph{Agreement Property.} We now use the previous two lemmas to prove the agreement property of consensus in Theorem \ref{theorem:agreement}.

\begin{theorem}
\label{theorem:agreement}
  (Agreement)
  If a correct process $p_i$ decides some value $v$, then no correct process $p_j$ decides a value $v'$ such that $v' \neq v$.
\end{theorem}

\pf\
Let us assume without loss of generality that some correct process $p_i$ is the first to decide, and decides value $v$ during some round $r$. From Lemma \ref{lemma:prepare-same-round} we can deduce that no correct process can decide $v'$ during $r$ because no correct process prepares for $v'$ during $r$.

For any round $r'$ such that $r' > r$, it follows from Lemma \ref{lemma:no-preprepare-following-rounds} that no correct process can decide $v'$ during $r'$ because there is no \preprepare{\lambda}{r'}{v'} message $m$ such that \textproc{JustifyPreprepare}($m$) is true. Therefore, no \preprepare{\lambda}{r'}{v'} message is accepted by any correct process and the algorithm does not make progress towards a decision on $v'$.

%%%%%%%%%%%%%%%%%%%%%%%%%%%%%%%%%%%%%%%%%%%%%%%%%%%%%%%%%%%%%%%%%%%%%
\subsection{Validity}

\begin{theorem}
	(Validity) Given an externally provided predicate $\beta$, if a correct process decides some value $v$, then $\beta(v)$ is true.	
\end{theorem}

\pf\ To decide, a correct process must receive a quorum of valid \commit{\lambda}{r}{v} messages. A message is only considered valid if it carries a proposal value $v$ such that $\beta(v)$ is true. As such, if a correct process decides $v$, then $\beta(v)$ must be true.

%%%%%%%%%%%%%%%%%%%%%%%%%%%%%%%%%%%%%%%%%%%%%%%%%%%%%%%%%%%%%%%%%%%%%%

\subsection{Termination}

\begin{theorem}
	(Termination) Every correct process decides.
\end{theorem}

\begin{proof}
  \beforePfSpace{1ex, 1ex}
  \afterPfSpace{1ex, 2ex}
  \interStepSpace{1ex}	
  \pflongnumbers

  For this proof, we assume, without loss of generality, that a correct process $p_i$ has not decided yet and demonstrate that it eventually decides.
  
  \vspace{2.0ex}
  \step{1}{
    $p_i$ must eventually reach some round $r$ (i.e., it sets $r_i = r$ and broadcasts a \roundchange{\lambda}{r}{_-}{_-} message) after GST with a correct leader $p_L$.
  }
  \begin{proof}
    \pf\ Since by assumption $p_i$ does not decide, its round timer must keep expiring indefinitely until it reaches round $r$ or it receives $f+1$ \texttt{ROUND-CHANGE} messages where $r$ is the highest round.
  \end{proof}

We now have two cases to consider. One where some correct process $p_j$ has already decided some value $v$, and one where no correct process has decided yet.
\vspace{2ex}

  \step{2}{
    \case{Some correct process $p_j$ has decided.}
    \begin{proof}
      \step{2.1}{
        $p_j$ receives the \roundchange{\lambda}{r}{_-}{_-} message broadcast by $p_i$ and sends to $p_i$ a quorum of valid COMMIT messages for the same value $v$.
      }
      \begin{proof}
        \pf\ By \stepref{1}, since after GST all messages sent by correct processes are timely delivered.
      \end{proof}
      
      \qedstep
      \begin{proof}
        \pf\ After GST, $p_i$ must receive the quorum of valid COMMIT messages for value $v$ sent $p_j$ and decide.
      \end{proof}
    \end{proof}

  }

  \step{3}{
    \case{No correct process has decided.}
    \begin{proof}
      \step{3.1}{
        Every correct process eventually reaches round $r$ and broadcasts a valid \roundchange{\lambda}{r}{_-}{_-} message.
      }
      \begin{proof}
        \pf\ The argument is the same as in step \stepref{1}, but generalized for all correct processes.
      \end{proof}
            
      \step{3.2}{
        The correct leader $p_L$ broadcasts a valid and justified \preprepare{\lambda}{r}{v} message.
      }
      \begin{proof}
        \pf\ After GST, $p_L$ must receive every \roundchange{\lambda}{r}{_-}{_-} message broadcast by a correct process. Since those messages are piggybacked with the necessary \texttt{PREPARE} messages for justification, $p_L$ is able to construct a valid and justified \preprepare{\lambda}{r}{v} message.
      \end{proof}
      
      \qedstep
      \begin{proof}
        \pf\ Since it is after GST, the algorithm will follow through with its normal case operation and $p_i$ will decide $v$ at the end of round $r$.
      \end{proof}
    \end{proof}
  }
\end{proof}

%%%%%%%%%%%%%%%%%%%%%%%%%%%%%%%%%%%%%%%%%%%%%%%%%%%%%%%%%%%%%%%%%%%%%%

\section{Conclusion}
\label{section:conclusion}
In this paper we proposed IBFT, a simple algorithm for the consensus problem in Byzantine-fault tolerant systems that is implemented in the Quorum blockchain. IBFT assumes a partially synchronous model. With timely communication, IBFT terminates within three message delays and has $O(n^2)$ communication complexity during both normal case operation and round changes. To ensure safety across round changes, IBFT relies on a  justification mechanism of proposed values that is critical in achieving quadratic communication complexity.

%We also present LinearIBFT, a variant that has worst-case $O(n)$ communication complexity and terminates within 5 message delays. LinearIBFT is the BFT consensus algorithm, within its class, with linear communication complexity that has lowest number of message delays.

\section*{Acknowledgements}
We would like to thank the following people for reviewing and  providing thoughtful feedback on earlier versions of this manuscript: Alysson Bessani and Rodrigo Rodrigues at the University of Lisbon, Jitendra Bhurat and Samer Falah at J.P. Morgan, and Roberto Saltini at ConsenSys.

\bibliography{hmz}
\bibliographystyle{abbrv}

%\appendix
%\section{TLA+ Specification}
%\label{section:specification}

%To be externally released at a later time.
%We have a TLA+ specification available on the team's GitHub repository with the name \texttt{ibft\_correct\_responsive.tla}.

\end{document}